\documentclass[11pt,reqno]{amsart}
\usepackage{amsmath,amsthm,amssymb,bm,bbm,dsfont,braket,comment,cite}
\usepackage{mathtools}\mathtoolsset{centercolon}
\mathtoolsset{showonlyrefs}
\usepackage[mathscr]{eucal}
\usepackage{amsaddr}
\usepackage{upgreek}
\usepackage[left=2cm,right=2cm,top=2cm,bottom=2cm]{geometry}
\usepackage{setspace}
\usepackage{graphicx}
\usepackage{verbatim}
\usepackage{float}
\usepackage{placeins}
\usepackage{array}
\usepackage{booktabs}
\usepackage{threeparttable}
\usepackage[update,prepend]{epstopdf}
\usepackage{multirow}
\usepackage{amsfonts,amssymb,dsfont,xfrac}
\usepackage[abs]{overpic}

\usepackage[usenames,dvipsnames]{color}
\usepackage[hidelinks]{hyperref}
\hypersetup{
	unicode=false,          
	pdftoolbar=true,        
	pdfmenubar=true,        
	pdffitwindow=false,     
	pdfstartview={FitH},    
	pdftitle={My title},    
	pdfauthor={Author},     
	pdfsubject={Subject},   
	pdfcreator={Creator},   
	pdfproducer={Producer}, 
	pdfkeywords={keyword1} {key2} {key3}, 
	pdfnewwindow=true,      
	colorlinks=true,        
	linkcolor=Red,          
	citecolor=ForestGreen,  
	filecolor=Magenta,      
	urlcolor=BlueViolet,    
}
\usepackage{doi}
\usepackage{url}
\usepackage{caption, subcaption}
\usepackage{enumitem}

\makeatletter
\ifx\@NODS\undefined%

\let\mathbb=\mathds
\else%
\fi
\makeatother

\def\d{{\text {\rm d}}}

\DeclareMathOperator{\Tr}{Tr}
\DeclareMathOperator{\e}{\mathrm{e}}

\DeclareMathOperator{\F}{\mathcal{F}}

\newcommand{\si}{{\sigma}}

\newcommand{\be}{{\mathbf e}}

\newcommand{\tr}{\operatorname{Tr}}
\newcommand{\al}{{\alpha}}

\newcommand{\ten}{\otimes}
\newcommand{\pl}{\hspace{.1cm}}

\newcommand{\norm}[2]{\parallel   #1   \parallel_{#2}}

        \def\cH{{\cal H}}

\def\0{{\mathbf{0}}}
\def\1{{\mathbf{1}}}
\def\2{{\mathbf{2}}}
\def\3{{\mathbf{3}}}
\def\4{{\mathbf{4}}}
\def\5{{\mathbf{5}}}
\def\6{{\mathbf{6}}}

\def\7{{\mathbf{7}}}
\def\8{{\mathbf{8}}}
\def\9{{\mathbf{9}}}

\def\be{\begin{equation}}
\def\ee{\end{equation}}
\def\bea{\begin{eqnarray}}
\def\eea{\end{eqnarray}}

\def\cH{{\mathcal H}}

\newcommand{\id}{\operatorname{id}}

\theoremstyle{plain}
\newtheorem{theo}{Theorem} 
\newtheorem{prop}[theo]{Proposition} 
\newtheorem{lemm}[theo]{Lemma} 

\theoremstyle{definition}

\theoremstyle{remark}
\newtheorem{remark}{Remark}[section]

\numberwithin{equation}{section}

\makeatletter
\newcommand{\opnorm}{\@ifstar\@opnorms\@opnorm}
\newcommand{\@opnorms}[1]{%
	\left|\mkern-1.5mu\left|\mkern-1.5mu\left|
	#1
	\right|\mkern-1.5mu\right|\mkern-1.5mu\right|
}
\newcommand{\@opnorm}[2][]{%
	\mathopen{#1|\mkern-1.5mu#1|\mkern-1.5mu#1|}
	#2
	\mathclose{#1|\mkern-1.5mu#1|\mkern-1.5mu#1|}
}
\makeatother

\begin{document}

\let\origmaketitle\maketitle
\def\maketitle{
	\begingroup
	\def\uppercasenonmath##1{} 
	\let\MakeUppercase\relax 
	\origmaketitle
	\endgroup
}

\title{\bfseries \Large{ Error Exponent and Strong Converse for Quantum Soft Covering
		}}

\author{ \normalsize \textsc{Hao-Chung Cheng$^{1,2,3,4}$ and Li Gao$^{5}$}}
\address{\small  	
	$^{1}$Department of Electrical Engineering and Graduate Institute of Communication Engineering,\\ National Taiwan University, Taipei 106, Taiwan (R.O.C.)\\
	$^{2}$Department of Mathematics, National Taiwan University\\
	$^{3}$Center for Quantum Science and Engineering,  National Taiwan University\\
	$^{4}$Hon Hai (Foxconn) Quantum Computing Center, New Taipei City 236, Taiwan\\
	$^{5}$Department of Mathematics, University of Houston, Houston, TX 77204, USA\\
}

\email{\href{mailto:haochung.ch@gmail.com}{haochung.ch@gmail.com}}
\email{\href{mailto:lgao12@uh.edu}{lgao12@uh.edu}}
\date{\today}
%
%
%
\maketitle

\begin{abstract}
How well can we approximate a quantum channel output state using a random codebook with a certain size? In this work, we study the quantum soft covering problem. Namely, we use a random codebook with codewords independently sampled from a prior distribution and send it through a classical-quantum channel to approximate the target state.
When using a random independent and identically distributed codebook with a rate above the quantum mutual information, we show that the expected trace distance between the codebook-induced state and the target state decays with exponent given by the sandwiched R\'enyi information. On the other hand, when the rate of the codebook size is below the quantum mutual information, the trace distance converges to one exponentially fast.
We obtain similar results when using a random constant composition codebook, whereas the sandwiched Augustin information expresses the error exponent.
In addition to the above large deviation analysis, our results also hold in the moderate deviation regime.
That is, we show that even when the rate of the codebook size approaches the quantum mutual information moderately quickly, the trace distance still vanishes asymptotically.

\end{abstract}

\section{Introduction} \label{sec:introduction}
Consider a classical-quantum (c-q) channel $\mathcal{N}_{X\to B}: x\mapsto \rho_B^x$ that takes every input letter $x\in\mathcal{X}$ to an output quantum state $\rho_B^x$ on a Hilbert space $\mathcal{H}_B$.
Given a probability distribution $p_X$ on the input alphabet $\mathcal{X}$ as an input to the channel, the corresponding output is then given by the marginal state $\rho_B = \sum_{x\in\mathcal{X}}p_X(x)\rho_B^x$.
Suppose we do not directly possess $p_X$ but have a random codebook $\mathcal{C}\subset \mathcal{X}$, in which the codewords are independently sampled from $p_X$. It is natural to expect that for large size of $\mathcal{C}$, the marginal state $\rho_B$ can be approximated by the codebook-induced output state defined by
\begin{align}
	\rho_B^\mathcal{C} := \frac{1}{|\mathcal{C}|} \sum_{x\in\mathcal{C}} \rho_B^x.
\end{align}
Namely, $\rho_B^\mathcal{C}$ is generated by uniformly choosing codewords in $\mathcal{C}$ and passing through the channel $\mathcal{N}_{X\to B}$.
Since the codebook $\mathcal{C}$ here is random, we take the expected value of the trace distance to quantify how well the induced state $\rho_B^\mathcal{C}$ approximates the true marginal state $\rho_B$, i.e.
\begin{align} \label{eq:metric}
	\varepsilon(\mathcal{C}):=\frac12\mathds{E}_\mathcal{C} \left\| \rho_B^\mathcal{C} - \rho_B \right\|_1.
\end{align}

When the underlying channel is classical, such a problem is called \emph{soft covering} and has been actively investigated due to its ample applications in secrecy analysis and some coding problems \cite{Wyn75b, HV93, Hay06sc, Cuf13, HK14, WH14, Cuf16,LCV16, PTM17, YT19, YC19, Yas19}.
In the quantum scenario, this problem was simply termed as ``\emph{covering}", and has been studied in the context of identification, compression, and channel simulation \cite{AW02, DW03, DW05, Win05, CWY04, LD09, BDH+14},  \cite[\S 17]{Wilde17}.
In this work, we study, for different types of random codebooks, how fast the trace distance in \eqref{eq:metric} converges to $0$ or to $1$, when the size of the random codebook is fixed. 
We term this study the \emph{large deviation analysis} for quantum soft covering \cite{YC19, Yas19}.

Firstly,
we establish a one-shot achievability bound (Theorem~\ref{theo:achievability_iid_oneshot}) and a one-shot strong converse bound (Theorem~\ref{theo:sc}) on the trace distance \eqref{eq:metric}, respectively.
These results directly apply to the $n$-shot extension with product channel $\mathcal{N}_{X\to B}^{\otimes n}$.
We prove that for an independent and identically distributed (i.i.d.) random codebook $\mathcal{C}^n$ with rate  $R:= \frac1n \log |\mathcal{C}^n|$, whose codewords being independently sampled from $p_X^{\otimes n}$, the following hold for every $n\in\mathds{N}$ (Propositions~\ref{prop:achievability_iid} and \ref{prop:sc_iid}),
\begin{align} \label{eq:iid}
	\begin{dcases}
		\frac12\mathds{E}_{{\mathcal{C}}^n} \left\|{\rho}_{B^n}^{{\mathcal{C}}^n}  -  {\rho}_{B}^{\otimes n} \right\|_1
		\leq
		\mathrm{e}^{  - n \mathsf{E}^*(R) }, &     R>I(X{\,:\,}B)_\rho \\
		\frac12\mathds{E}_{{\mathcal{C}}^n} \left\|{\rho}_{B^n}^{{\mathcal{C}}^n}  -  {\rho}_{B}^{\otimes n} \right\|_1 \geq 1  -  4\mathrm{e}^{-n \mathsf{E}_\text{sc}^\downarrow (R)}, &     R<I(X{\,:\,}B)_\rho \\
	\end{dcases}	
\end{align}
where the exponents $\mathsf{E}^*(R) := \sup_{\alpha \in (1,2)} \frac{1-\alpha}{\alpha} ( {I}_\alpha^{*} \left( X; B \right)_\rho  -  R )$ and $\mathsf{E}_\text{sc}^\downarrow (R)  := \sup_{\alpha \in (\frac12,1)} \frac{1-\alpha}{\alpha}( I_{2-\sfrac{1}{\alpha}}^{\downarrow}(X{\,:\,}B)_\rho  -  R)$ are defined in terms of the \emph{order-$\alpha$ sandwiched R\'enyi information} and a variant of the Petz-type R\'enyi information (see the detailed definitions in Section~\ref{sec:notation}). Our results hence imply that the quantum mutual information $I(X{\,:\,}B)_\rho$ is the \emph{minimal achievable rate} as well as the \emph{strong converse rate} for quantum soft covering using random i.i.d.~codebook.

Secondly, we consider a \emph{constant composition random codebook} where its codewords are independently sampled uniformly from the type class of $p_X$, whose distribution is 
\begin{align}
\breve{p}_{X^n} (x^n):= \frac{1}{|T_p^n|} \mathbf{1}_{ \{x^n\in T_p^n\}},
\end{align}
and the type class $T_p^n$
 is the set of all sequence $x^n$ with empirical distribution $p_X$. We write
\begin{align}
	\breve{\rho}_{B^n} := \mathcal{N}_{X\to B}^{\otimes n}(\breve{p}_{X^n})
\end{align}
as the corresponding channel output state. We show that, for every $n\in\mathds{N}$ (Theorem~\ref{theo:achievability_composition} and Proposition~\ref{prop:sc_cc}),
\begin{align} \label{eq:cc}
	       \begin{dcases}
		\frac12\mathds{E}_{\breve{\mathcal{C}}^n} \left\|{\rho}_{B^n}^{\breve{\mathcal{C}}^n}  -  \breve{\rho}_{B^n} \right\|_1
		\leq
		\mathrm{e}^{  - n \breve{\mathsf{E}}^*(R) }, &      R>I(X{\,:\,}B)_\rho \\
		\frac12\mathds{E}_{\breve{\mathcal{C}}^n} \left\|{\rho}_{B^n}^{\breve{\mathcal{C}}^n}  -  \breve{\rho}_{B^n} \right\|_1 \geq 1  -  n^{k_p}\mathrm{e}^{-n \breve{\mathsf{E}}_\text{sc}^\downarrow (R)}, &      R<I(X{\,:\,}B)_\rho \\
	\end{dcases}	
\end{align}
where the exponents ${\textstyle\breve{\mathsf{E}}^*(R):=\sup_{\alpha \in (1,2)}\frac{1-\alpha}{\alpha} (\breve{I}_\alpha^{*} (X{\,:\,}B)_\rho-R)}$ and $\breve{\mathsf{E}}_\text{sc}^\downarrow (R):= \sup_{\alpha \in (\sfrac12,1)} \frac{1-\alpha}{\alpha}( \breve{I}_{2-\sfrac{1}{\alpha}}^{\downarrow}(X{\,:\,}B)_\rho - R)$ are defined in terms of the \emph{order-$\alpha$ sandwiched Augustin information} and a variant of the Petz-type Augustin information.
Again, $I(X{\,:\,}B)_\rho$ acts as the fundamental limit for the minimal achievable rate.
However, notably both the exponents when using the random constant composition codebook are larger than that of using the random i.i.d.~ codebook, which indicates a faster convergence.
We remark that both our results established in \eqref{eq:iid} and \eqref{eq:cc} hold for every finite blocklength $n$.

Our result extends to the moderate deviation regime. Namely, as the rate $R$ of both the random codebooks approaches $I(X{\,:\,}B)_\rho$ from above at a speed no faster than $O(\sfrac{1}{\sqrt{n}})$, the trace distances still vanishes asymptotically\footnote{Here, by ``$f(n)\lesssim g(n)$" we meant $\lim_{n\to\infty} \frac{1}{n a_n^2} \log f(n) \leq \lim_{n\to\infty} \frac{1}{n a_n^2} \log g(n) $. See Propositions~\ref{prop:moderate_iid} and \ref{prop:moderate_cc} for the precise statements.}  (Propositions~\ref{prop:moderate_iid} and \ref{prop:moderate_cc}):
\begin{align}
	\frac12\mathds{E}_{{\mathcal{C}}^n} \left\|{\rho}_{B^n}^{{\mathcal{C}}^n}  -  {\rho}_{B}^{\otimes n} \right\|_1 
	&\lesssim \mathrm{e}^{-\frac{na_n^2}{2V(X{\,:\,}B)_\rho}  } \to 0, \quad   R = I(X{\,:\,}B)_\rho + a_n; \\
	\frac12\mathds{E}_{\breve{\mathcal{C}}^n} \left\|{\rho}_{B^n}^{\breve{\mathcal{C}}^n}  -  \breve{\rho}_{B^n} \right\|_1
	&\lesssim \mathrm{e}^{-\frac{na_n^2}{2\breve{V}(X{\,:\,}B)_\rho}  } \to 0, \quad   R = I(X{\,:\,}B)_\rho + a_n. 
\end{align}
Here, $V(X{\,:\,}B)_\rho$ is the quantum information variance, $\breve{V}(X{\,:\,}B)_\rho$ is a variant of it, and $(a_n)_{n\in\mathds{N}}$ is any moderate deviation sequence satisfying $a_n\downarrow 0$ and $n a_n^2 \uparrow \infty$.

Our lower estimates in trace distance can be compared to the covering lemma in \cite[\S 17]{Wilde17} (see also \cite{DW03}), which proves that for any $R>I(X{\,:\,}B)_{\rho}$ and $\delta>0$, the probability of $\varepsilon(\mathcal{C})\ge \delta$ converges to $0$ in probability.
Note that our upper estimates in \eqref{eq:iid} and \eqref{eq:cc} for the trace norm also implies exponential convergence in probability. 

This paper is organized as following.
Section~\ref{sec:notation} presents necessary notation and information quantities.
In Section~\ref{sec:achievability}, we prove the achievability (i.e.~exponential upper bound), and in Section~\ref{sec:sc}, we prove the exponential strong converse.
Section~\ref{sec:moderate} presents moderate deviation analysis.
We conclude this paper in Section`\ref{sec:conclusions}.
Appendix~\ref{sec:interpolation} introduces basics of the complex interpolation theory.

\section{Notation and Information Quantities} \label{sec:notation}
For a Hilbert space $\mathcal{H}$, we denote $\mathcal{B(H)}$ and $\mathcal{B}_{\geq 0}(\mathcal{H})$ the set of bounded linear operators and the set of positive semi-definite operators on $\mathcal{H}$.
The set of density operators $\mathcal{S(H)}$ is positive semi-definite operators with unit trace.
For $p\geq 1$, the Schatten $p$-norm is
	$\left\|M\right\|_{S_p(\mathcal{H}) } = \left( \Tr\left[ | M|^p \right] \right)^{\sfrac1p}$.
The set of bounded linear operators with finite Schatten $p$-norm is denoted as the Schatten $p$-class $S_p(\mathcal{H})$.
We will often shorthand $\|\cdot\|_p \equiv \|\cdot\|_{S_p(\mathcal{H}) }$ if the underlying Hilbert space is clear and there is possibility of confusion.
We use $\texttt{supp}(\cdot)$ to stand for the support of an operator or the support of a function.
We use $\mathds{E}_{x\sim p_X}$ to denote taking expectation with respect to random variable $x$ governed by probability distribution $p_X$, e.g..~
\begin{align}
	\mathds{E}_{x\sim p_X} \left[|x\rangle \langle x|\otimes \rho_B^x \right]
	= \sum_{x\in\mathcal{X}} p_X(x) |x\rangle \langle x|\otimes \rho_B^x = \rho_{XB}.
\end{align}

We define the order-$\alpha$ Petz--R\'enyi divergence $D_\alpha$ \cite{Pet86} and the sandwiched R\'enyi divergence $D^*_\alpha$ \cite{MDS+13,WWY14} for $\rho \in \mathcal{S(H)}$ and $\sigma \in \mathcal{B}_{\geq 0}(\mathcal{H})$ and $\alpha\in(0,\infty)\backslash 1$ as
\begin{align} D_\alpha(\rho\|\sigma)&:=\frac{1}{\al-1}\log\tr\left[\rho^{\alpha}\sigma^{1-\alpha}\right]\pl, \\
D_\alpha^*(\rho\|\sigma)&:=\frac{1}{\al-1}\log\norm{\sigma^{\frac{1-\alpha}{2\alpha}}\rho\si^{\frac{1-\alpha}{2\alpha}}}{\al}^\al.
\end{align}
Note that both R\'enyi divergences converge to the \emph{quantum relative entropy} \cite{Ume62},\cite[Lemma 3.5]{MO14}, i.e.~
\begin{align}
	\lim_{\alpha\to 1} D_\alpha(\rho\|\sigma) = \lim_{\alpha\to 1} D_\alpha^*(\rho\|\sigma) = D(\rho\|\sigma) := \Tr\left[ \rho (\log \rho - \log \sigma ) \right].
\end{align}
We define the \emph{relative entropy variance} $V(\rho\|\sigma)$ is defined by
\[
V(\rho\|\sigma): = \Tr\left[\rho(\log \rho -\log \sigma )^2\right] - \left( D(\rho\|\sigma) \right)^2.
\]

For a classical-quantum (c-q) state $\rho_{XB} = \sum_{x\in\mathcal{X}} p_X(x) |x\rangle \langle x| \otimes \rho_B^x$, we define the \emph{order-$\alpha$ sandwiched R\'enyi information} $I_\alpha^{*} \left( X{\,:\,}B \right)_\rho$ and the \emph{order-$\alpha$ sandwiched Augustin information} $\breve{I}_\alpha^{*} \left( X{\,:\,}B \right)_\rho$ as following:
\begin{align}
	\begin{split} \label{eq:sandwiched_Renyi}
	I_\alpha^{*} \left( X{\,:\,}B \right)_\rho &:= \inf_{\sigma_B\in\mathcal{S}(\mathcal{H}_B)} D_\alpha^*\left( \rho_{XB} \| p_X \otimes \sigma_B \right) \\
	&= \inf_{\sigma_B\in\mathcal{S}(\mathcal{H}_B)} \frac{\alpha}{\alpha-1} \log \left( \sum_{x\in\mathcal{X}} p_X(x) \left\| \sigma_B^{\frac{1-\alpha}{2\alpha} }  \rho_B^x \sigma_B^{\frac{1-\alpha}{2\alpha}} \right\|_\alpha^\alpha \right)^{\frac{1}{\alpha}};
	\end{split} \\
	\begin{split} \label{eq:sandwiched_Augustin}
	\breve{I}_\alpha^{*} \left( X{\,:\,}B \right)_\rho &:= \inf_{\sigma_B\in\mathcal{S}(\mathcal{H}_B)} \sum_{x\in\mathcal{X}} p_X(x) D_\alpha^*\left( \rho_B^x \| \sigma_B \right) \\
	&= \inf_{\sigma_B\in\mathcal{S}(\mathcal{H}_B)} \frac{\alpha}{\alpha-1}  \sum_{x\in\mathcal{X}} p_X(x) \log \left\| \sigma_B^{\frac{1-\alpha}{2\alpha} }  \rho_B^x \sigma_B^{\frac{1-\alpha}{2\alpha}} \right\|_\alpha.
	\end{split}
\end{align}
Here the infimum $\si_B$ is taken over all densities on $B$.
Moreover, we define the following variants of the Petz-type information quantities:
\begin{align}
	I^{\downarrow}_{\alpha}(X{\,:\,}B)_\rho
	&:=  D_\alpha\left( \rho_{XB} \| \rho_X \otimes \rho_B \right); \label{eq:Petz_Renyi_down}\\
	\breve{I}^{\downarrow}_{\alpha}(X{\,:\,}B)_\rho
	&:= \sum_{x\in\mathcal{X}} p_X(x) D_\alpha\left( \rho_{B}^x \|  \rho_B \right). \label{eq:Petz_Augustin_down}
\end{align}
All the fourth information quantities converges to the \emph{ quantum mutual information}, i.e.
\begin{align} \label{eq:alpha1}
	\begin{split}
\lim_{\alpha\to 1} I_\alpha^{*} \left( X{\,:\,}B \right)_\rho
&=\lim_{\alpha\to 1}\breve{I}_\alpha^{*} \left( X{\,:\,}B \right)_\rho
=\lim_{\alpha\to 1}I^{\downarrow}_{\alpha}(X{\,:\,}B)_\rho
=\lim_{\alpha\to 1}\breve{I}^{\downarrow}_{\alpha}(X{\,:\,}B)_\rho \\
&= I(X{\,:\,}B)_\rho
:= D(\rho_{XB}\|\rho_X\otimes \rho_B).
\end{split}
\end{align}
For a c-q state $\rho_{XB} = \sum_{x\in\mathcal{X}} p_X(x) |x\rangle \langle x| \otimes \rho_B^x$, we define the \emph{mutual information variance} $V(X{\,:\,}B)_{\rho}$ and a variant $\breve{V}(X{\,:\,}B)_{\rho}$ as
\begin{align}
V(X{\,:\,}B)_{\rho} &:= V(\rho_{XB} \,\|\,\rho_X\otimes \rho_B); \\
\breve{V}(X{\,:\,}B)_{\rho} &:= \mathds{E}_{x\sim p_X} \left[ V(\rho_B^x\,\|\,\rho_B )\right].
\end{align}

We remark that both the quantities introduced in \eqref{eq:sandwiched_Renyi} and \eqref{eq:sandwiched_Augustin} do not have a closed-form expression. However, an iterative optimization algorithm with convergence guarantees has been proposed to compute them \cite{YCL21}.

\section{Achievability} \label{sec:achievability}

Let 
$\rho_{XB} := \sum_{x\in\mathcal{X}} p_X(x) |x\rangle \langle x| \otimes \rho_B^x$ be a classical-quantum state.
The goal of quantum soft covering is to approximate the marginal state at the channel output, i.e.~$\rho_B = \sum_{x\in\mathcal{X}} p_X(x)  \rho_B^x$, given access to the classical-quantum channel $x\mapsto \rho_B^x$ and sampling from the prior distribution $p_X$.

To that end, we consider a random codebook $\mathcal{C} = \left\{ x(m)  \right\}_{m=1}^M \subseteq \mathcal{X}$ of size $M$, where its codewords $x(1),\cdots,x(m)$ are independently generated according to  $p_X$.
Then, the average state induced by the random codebook~$\mathcal{C}$ is:
\begin{align}
	{\rho}_B^{\mathcal{C}}  := \frac1M \sum_{x\in\mathcal{C}} \rho_B^{x}.
\end{align}
Hence, we take the expected value (over the random codebook $\mathcal{C}$) of the trace distance between the codebook-induced state ${\rho}_B^{\mathcal{C}}$ and the true marginal state $\rho_B$ as the figure of merit:
\begin{align}
	\frac12 \mathds{E}_{\mathcal{C}} \left\| {\rho}_B^{\mathcal{C}} - \rho_B \right\|_1.
\end{align}
The main result of this section is to prove the following upper bound on the trace distance when the codebook size $M$ is fixed.
\begin{theo}[A one-shot achievability via R\'enyi Information] \label{theo:achievability_iid_oneshot}
	The trace distance between the induced state $\rho_{B}^{\mathcal{C}}$ and the true state $\rho_B$ is upper bounded by
	\begin{align}
		\frac12\mathds{E}_\mathcal{C} \left\|\rho_{B}^{\mathcal{C}} - \rho_B \right\|_1
		\leq 2^{\frac{2}{\alpha}-2}
		\e^{ \frac{\alpha-1}{\alpha } \left( I_\alpha^{*} \left(X{\,:\,}B \right)_\rho - \log M \right)}, \quad \al\in (1,2). 
	\end{align}
	Here, the order-$\alpha$ sandwiched R\'enyi information $I_\alpha^{*} \left(X{\,:\,}B \right)_\rho$ is defined in \eqref{eq:sandwiched_Renyi}.
\end{theo}
This one-shot achievability bound applies to the $n$-shot scenario when using a \emph{random independent and identically distributed (i.i.d.) codebook} (Section~\ref{sec:achievability_iid})
and a \emph{random constant composition codebook} (Section~\ref{sec:achievability_cc}).

To prove achievability on random codebook, we start with a lemma to exploit the independence between random codewords. Let $L_\infty(\Omega,\mu)$ be a probability space. For $M\ge 1$, we write $\Omega^{ M}={\Omega \times \cdots\times\Omega}
$ for an $M$-fold product space of $\Omega$.  For $1\le i\le M$, we define the following maps:
\begin{align*}&\pi_i: L_\infty(\Omega,\mu)\to L_\infty\left(\Omega^{ M},\mu\right) \pl, \pl \\
	&\pi_i(f)(\omega_1,\dots,\omega_M)=f(\omega_i) \pl, \pl \\
	&E: L_\infty(\Omega,\mu)\to L_\infty(\Omega^{ M},\mu)\pl, \pl\\
	&E(f)(\omega_1,\dots,\omega_M)=\int_{\Omega} f(\omega) \d\mu(\omega)=:\mathbb{E}_\mu(f)  \pl, \pl \\
	&\Theta:= \frac{}{}: L_\infty(\Omega,\mu)\to L_\infty\left(\Omega^{ M},\mu\right)\pl, \pl \\
	&\Theta(f)=\frac{1}{M}\sum_{i=1}^M \pi_i(f)-E(f),
\end{align*}
where $(\omega_1,\cdots, \omega_M)\in \Omega^{\times M}$.
Here, $\pi_i$ is an embedding such that $\pi_i(f)$ only depends on the $i$-th coordinate $\omega_i$ via $f$, and $\mathbb{E}$ sends $f$ to the constant function of its mean $\mathbb{E}_\mu(f)$. It is clear to see that $\pi_i(f)$ forms an i.i.d.~copy of distribution of $f$. Our key lemma in achievability is to upper bounds the norm of the operation $\Theta$ on operator-valued functions.
\begin{lemm}\label{lemm:inter}Let $\Theta$ be the map defined above. Then, for any Hilbert space $\cH$ and $1\le p\le 2$, we have
	\[\norm{\Theta\ten \id: L_p(\Omega, S_p(\cH))\to L_p\left(\Omega^{M}, S_p(\cH)\right)}{}\le 2^{\frac{2}{p}-1}M^{\frac{1-p}{p}},  \]
	where the identity map $\id$ is acting on $S_p(\cH)$.
\end{lemm}
\begin{proof}
	For any $p \ge 1$, it is clear that for each $i$, $\pi_i$ gives an isometry on $L_p$-spaces
	\begin{align*}
		&\norm{\pi_i\ten \id (f)}{L_p(\Omega^{\times M}, S_p(\cH))}^p\\
		&=
		\int_{\Omega^{M} }\norm{\pi_i\ten \id (f)(\omega_1,\cdots ,\omega_M)}{p}^p \d\mu(\omega_1) \cdots  \d\mu(\omega_M)\\
		&=
		\int_{\Omega }\norm{f(\omega_i)}{p}^p   \d\mu(\omega_i)
		=\norm{f}{L_p(\Omega, S_p(\cH))}^p.
	\end{align*}
	Moreover, by convexity of $S_p(\cH)$ norm, the map $E$ is a contraction, i.e.~
	\begin{align*}  &\norm{E(f)}{L_p\left(\Omega^{ M}, S_p(\cH)\right)}=\norm{\mathbb{E}_{\mu}(f)}{ S_p(\cH)}\\
		&\le  \int_{\Omega}\norm{f(\omega)}{p}\d\mu(\omega)\le \Big(\int_{\Omega}\norm{f(\omega)}{p}^p\d\mu(\omega)\Big)^{\sfrac{1}{p}}\pl,
	\end{align*}
	where in the last step we used Jensen's inequality and H\"older inequality on a probability space.
	Then by triangle inequality,  we have for $p=1$,
	\[ \norm{\Theta\ten \id: L_1(\Omega, S_1(\cH))\to L_1(\Omega^{M}, S_1(\cH))}{}\le 2\pl.\]
	For $p=2$, we consider
	\[ \Theta\ten \id(f)=\frac{1}{M}\sum_{i=1}^M (\pi_i(f)- E(f))=\frac{1}{M}\sum_{i=1}^M \hat{f}_i\pl,\]
	where $\hat{f}_i= \pi_i(f)- E(f)$ is the mean zero part of $\pi_i(f)$.
	It then follows from independence that $\hat{f}_i$ are mutually orthogonal (even in the operator-valued inner product). Indeed, for $i\neq j$,
	\begin{align*}\mathbb{E}(\hat{f}_i^* \hat{f}_j)=&\left(\mathbb{E} (\pi_i(f)-E(f))\right)^* \left(\mathbb{E}(\pi_j(f)-E(f))\right)\\
		=&|\mathbb{E}_{\mu}f-\mathbb{E}_{\mu}f|^2=0.\end{align*}
	This further implies that $\hat{f}_i$ are orthogonal in the Hilbert space $L_2(\Omega^{ M}, S_2(\cH))$. Note that for each $i$,
	\begin{align} \norm{\pi_i(f)-E(f)}{L_2(\Omega^{ M}, S_2(\cH))}
		&=\norm{\pi_i(f-\mathbb{E}_\mu f)}{L_2(\Omega^{ M}, S_2(\cH))}\\
		&=\norm{f-\mathbb{E}_\mu f}{L_2(\Omega, S_2(\cH))}\\
		&\le \norm{f}{L_2(\Omega, S_2(\cH))},
	\end{align}
	where the last inequality follows from the fact that $ f\to \mathbb{E}_\mu f$ is the projection from $L_2(\Omega, S_2(\cH))$ onto the (operator-valued) constant function. Thus, we have
	\begin{align*}
		\norm{\Theta\ten \id(f)}{2}^2&= \left\|\frac{1}{M}\sum_{i=1}^M \hat{f}_i\right\|_2^{2} \\
		&= \frac{1}{M^2}\sum_{i,j=1}^M \langle \hat{f}_i,\hat{f}_j \rangle\\
		&= \frac{1}{M^2}\sum_{i=1}^M \left\|\hat{f}_i\right\|^{2} \\
		&\le \frac{1}{M} \norm{f}{2}.
	\end{align*}
	This means that, for $p=2$,
	\[ \norm{\Theta\ten \id: L_2(\Omega, S_2(\cH))\to L_2\left(\Omega^{ M}, S_2(\cH)\right)}\le \frac{1}{\sqrt{M}}\pl.\]
	The case of general $1\le p\le 2$ follows from interpolation (see e.g. \cite{BL76}) with $\theta = \frac{2(p-1)}{p} \in [0,1]$.
\end{proof}
Since the above estimate does not depends on the dimension of $\mathcal{B(H)}$, in the following we  write $\Theta$ for $\Theta\ten \id_{\mathcal{B(H)}}$ if no confusion. Similarly, we omit the notation of $\id_{\mathcal{B(H)}}$ for the maps $\pi_i$ and $E$.
We are now ready to prove the one-shot achievability bound.

\begin{proof}[Proof of Theorem~\ref{theo:achievability_iid_oneshot}]
	For the ease of notation,
	we write $\mathcal{H}\equiv \mathcal{H}_B$,  $\rho_x\equiv\rho_B^{x}$ and $\rho_B\equiv \sum_{x}p_X(x)\rho_x$ throughout the proof.
	Let $x: \Omega\to \mathcal{X}$ be a random codeword with respect to the distribution $p_X$, where $\Omega$ is the event space. We can rewrite the classical-quantum state $\rho_{XB}$ as
	\[\rho_{\Omega B}= \sum_{x} 1_{A_{x}}\ten \rho_x \in L_\infty(\Omega, \mathcal{B(H)})\pl,\]
	where $1_{A_{x}}$ is the characteristic function on the mutually disjoint set $A_x$, which satisfying  $\Pr(A_x)=p_X(x)$ and $\sum_{x\in\mathcal{X}}\Pr(A_x)=\sum_{x\in\mathcal{X}} p_X(x)=1$. In particular, we have $\mathds{E}_{\Omega}\rho_{\Omega B}=\rho_B$.
	Take $1< \al< 2$ and $\frac{1}{\al}+\frac{1}{\al'}=1$. The R\'enyi information can be expressed as
	\begin{align*} I_\alpha^{*} \left( X{\,:\,}B \right)_\rho & =   \inf_{\sigma} \frac{1}{\alpha - 1} \log   \left( \sum_{x\in\mathcal{X}} p_X(x) \left\| \sigma^{-\frac{1}{2\al'} }  \rho_x\sigma^{-\frac{1}{2\al'} } \right\|_\alpha^\alpha \right)\\
		& =  \inf_{\sigma\in\mathcal{S}(\mathcal{H})} \frac{\alpha}{\alpha - 1} \log \norm{\sigma_{\Omega B}^{-\frac{1}{2\al'} }\rho_{\Omega B}  \sigma_{\Omega B}^{-\frac{1}{2\al'} }}{L_\al(\Omega,S_{\al}(\cH))}
	\end{align*}
	where $\sigma_{\Omega B}:=1_\Omega \ten \sigma$ is interpreted as a constant function on $L_\infty(\Omega,\mathcal{B}(\cH))$. In other words,
	\begin{align*}\e^{ \frac{\alpha-1}{\alpha } I_\alpha^{*} \left(X{\,:\,}B \right)_\rho}=\inf_{\sigma\in\mathcal{S}(\mathcal{H})} \norm{\sigma_{\Omega B}^{-\frac{1}{2\al'} }\rho_{\Omega B}  \sigma_{\Omega B}^{-\frac{1}{2\al'} }}{L_\al(\Omega,S_{\al}(\cH))}\pl.
	\end{align*}
	Now using the construction in the Lemma~\ref{lemm:inter}, we have
	\begin{align*}\mathds{E}_\mathcal{C} \left\|\rho_{B}^{\mathcal{C}} - \rho_B \right\|_1
		&=\mathds{E}_\mathcal{C} \left\|\frac1M \sum_{x\in\mathcal{C}} \rho_x - \rho_B \right\|_1\\
		&= \norm{\Theta (\rho_{\Omega B})}{L_1(\Omega^{ M},S_1(\mathcal{H}))}.
	\end{align*}
	Note that for any state $\sigma\in\mathcal{S}(\mathcal{H})$,
	\begin{align*}
		\norm{\Theta (\rho_{\Omega B})}{L_1(\Omega^{\times M},S_1(\mathcal{H}))}
		&\overset{\text{(a)}}{\le} \;\norm{\sigma^{-\frac{1}{2\al'} }\Theta (\rho_{\Omega B})\sigma^{-\frac{1}{2\al'} }}{L_\al(\Omega^{ M},S_\al(\mathcal{H}))}
		\\
		&\overset{\text{(b)}}{=} \; \norm{\Theta (\sigma_{\Omega B}^{-\frac{1}{2\al'} }\rho_{\Omega B}\sigma_{\Omega B}^{-\frac{1}{2\al'} })}{L_\al(\Omega^{ M},S_\al(\mathcal{H}))}
		\\
		&\le \;\norm{\Theta: L_\al(\Omega,S_\al(\mathcal{H}))\to L_\al(\Omega^{ M},S_\al(\mathcal{H}))}{} \cdot \norm{ \sigma_{\Omega B}^{-\frac{1}{2\al'} }\rho_{\Omega B}\sigma_{\Omega B}^{-\frac{1}{2\al'} }}{L_\al(\Omega^{ M},S_\al(\mathcal{H}))}.
	\end{align*}
	Here, (b) used the fact $\Theta= \Theta \ten \id_{\mathcal{B(H)}}$ is identity on the operator part , and (a) used H\"older inequality ${\norm{AXB}{1}\le \norm{A}{2\al'}\norm{X}{\al}\norm{B}{2\al'}}$ with
	\[\norm{ \sigma^{\frac{1}{2\al'} }}{L_\al(\Omega^{ M},S_\al(\mathcal{H}))}=\int_{\Omega^{\times M}} \norm{\sigma^{\frac{1}{2\al'} }}{2\al'}^{2\al'}d\mu^{ M}=1\]
	Then the assertion follows from Lemma~\ref{lemm:inter} and taking infimum over $\sigma \in \mathcal{S}(\mathcal{H})$, i.e.~
	\begin{align*}\mathds{E}_\mathcal{C} \left\|\rho_{B}^{\mathcal{C}} - \rho_B \right\|_1
		&\le  2^{\frac{2}{\alpha}-1}M^{\frac{1-\alpha}{\alpha}} \cdot  \inf_{\sigma\in\mathcal{S}(\mathcal{H})}  \norm{ \sigma_{\Omega B}^{-\frac{1}{2\al'} }\rho_{\Omega B}\sigma_{\Omega B}^{-\frac{1}{2\al'} }}{L_\al(\Omega^{\times M},S_\al(\mathcal{H}))}\\
		&= 2^{\frac{2}{\alpha}-1}M^{\frac{1-\alpha}{\alpha}}\e^{ \frac{\alpha-1}{\alpha } I_\alpha^{*} \left(X{\,:\,}B \right)_\rho}.
	\end{align*}
	That finishes the proof.
\end{proof}

\noindent\subsection{Random I.I.D.~Codebook} \label{sec:achievability_iid}
We consider the $n$-shot extension of quantum soft covering, where the c-q channel is now $n$-fold product:
\begin{align}
	x^n \mapsto \rho_{B^n}^{x^n} := \rho_B^{x_1}\otimes \cdots\otimes \rho_B^{x_n}, \quad \forall x^n \in \mathcal{X}^n.
\end{align}

In this section, we investigate the case that input distribution is i.i.d., i.e.~$p_{X^n} = p_X^{\otimes n}$.
Hence, the joint c-q state is the $n$-fold product state $\rho_{X^n B^n} = \rho_{XB}^{\otimes n}$ and our targeted true marginal state is $\rho_B^{\otimes n}$.

With the i.i.d.~prior $p_X^{\otimes n}$,
we use the codebook $\mathcal{C}^n$ with size $|\mathcal{C}^n| = \exp(nR)$, where
each codeword in $\mathcal{C}^n$ is i.i.d.~drawn according to $p_X^{\otimes n}$.
We term this the \emph{random i.i.d.~codebook}.
Now the goal is to use the codebook-induced state $
	\rho_{B^n}^{\mathcal{C}^n}$ to approximate $\rho_B^{\otimes n}$.

We apply the one-shot achievability established in Theorem~\ref{theo:achievability_iid_oneshot} to show that the expected value of the trace distance between the induced state ${\rho}_{B^n}^{{\mathcal{C}}^n}$ and the true marginal state ${\rho}_{B}^{\otimes n}$ decays exponentially fast.
\begin{prop}[$n$-shot achievability using random i.i.d.~codebook] \label{prop:achievability_iid}
	For any $n\in\mathbb{N}$, let $R= \frac1n \log |{\mathcal{C}}^n|$.
	Then,
	\begin{align}
		\frac12\mathds{E}_{{\mathcal{C}}^n} \left\|{\rho}_{B^n}^{{\mathcal{C}}^n} - {\rho}_{B}^{\otimes n} \right\|_1
		\leq
		\mathrm{e}^{  - n  \sup_{\alpha \in (1,2)} \frac{1-\alpha}{\alpha} \left( {I}_\alpha^{*} \left( X{\,:\,}B \right)_\rho - R \right) }.
	\end{align}
	Moreover, the exponent $\sup_{\alpha \in (1,2)} \frac{1-\alpha}{\alpha} \left( {I}_\alpha^{*} \left( X{\,:\,}B \right)_\rho - R \right)$ is positive if and only if $R> I(X{\,:\,}B)_\rho$.
\end{prop}

\begin{proof}
	Recall the additivity of sandwiched R\'enyi information \cite[Lemma 4.8]{MO14}:
	\begin{align} \label{eq:additivity}
		I_\alpha^{*} \left( X^n; B^n \right)_\rho = n I_\alpha^{*} \left( X{\,:\,}B \right)_\rho.
	\end{align}
	The exponential upper bound is a direct consequence of the one-shot achievability proved in Theorem~\ref{theo:achievability_iid_oneshot} and \eqref{eq:additivity}.
	The positivity follows from the monotone non-decreasing map $\alpha\mapsto I_\alpha^*$ \cite[Proposition~4]{CGH18} and \eqref{eq:alpha1}.
\end{proof}

\subsection{Random Constant Composition Codebook} \label{sec:achievability_cc}
Our Theorem \ref{theo:achievability_iid_oneshot} also applies to achievability using random constant composition codebook as well. Let $\rho_{XB} = \sum_{x\in\mathcal{X}} p_X(x) |x\rangle\langle x|\otimes \rho_B^x$ be a c-q state. Fix $n$ such that $n p_X(x) \in 0\cup\mathbb{N}$ for all $x\in\mathcal{X}$.
We define the type class of length-$n$ sequences under $p_X$ as
\begin{align}
	T_p^n := \set{ x^n \in \mathcal{X}^n : P_{x^n} = p_X	},
\end{align}
where the empirical distribution of sequence $x^n\in\mathcal{X}^n$ is
\begin{align} \label{eq:empirical}
	P_{{x}^n} (x) := \frac1n \sum_{i=1}^n \mathbf{1}_{\{ x = x_i\}}, \quad \forall x\in\mathcal{X}.
\end{align}
Then, we define a uniform distribution on the type class as:
\begin{align}
	\breve{p}_{X^n}(x^n):=\frac{1}{\left|T_{p}^n\right|} \mathbf{1}_{x^n \in T_{p}^n}, \quad x^n \in \mathcal{X}^n.
\end{align}
A \emph{random constant composition codebook}
\begin{align} \label{eq:composition_codebook}
	\breve{\mathcal{C}}^n:=\{x^n(m)\}_{m=1}^M, \quad x^n(m)\sim \breve{p}_{X^n}(x^n)
\end{align}
consists of $M$ codewords, where each codeword $x^n(m)$ is independently drawn according to distribution $\breve{p}_{X^n}$. For a c-q channel $x\mapsto \rho_B^x$ (induced by the c-q state $\rho_{XB}$), we define the c-q state generated from $\breve{p}_{X^n}$ as
\begin{align}
\breve{\rho}_{X^nB^n} := \sum_{x^n\in\mathcal{X}^n} \breve{p}_{X^n}(x^n) |x^n\rangle \langle x^n| \otimes \rho_{B^n}^{x^n},
\end{align}
 and its marginal state \[\breve{\rho}_{B^n}:= \sum_{x^n\in\mathcal{X}^n} \breve{p}_{X^n}(x^n) \rho_{B^n}^{x^n}\pl.\]
The induced output state via the random constant composition codebook $\breve{\mathcal{C}}^n$ is then
\begin{align}
	{\rho}_{B^n}^{\breve{\mathcal{C}}^n} := \frac{1}{|\breve{\mathcal{C}}^n|} \sum_{x^n\in\breve{\mathcal{C}}^n} \rho_{B^n}^{x^n}.
\end{align}

\begin{theo}[$n$-shot achievability using random constant composition codebook] \label{theo:achievability_composition}
	For any $n\in\mathbb{N}$, consider a c-q state $\rho_{XB} = \sum_{x\in\mathcal{X}} p_X(x) |x\rangle\langle x|\otimes \rho_B^x$, where $n p_X(x) \in 0\cup\mathbb{N}$ for all $x\in\mathcal{X}$, and let $R:= \frac1n \log |\breve{\mathcal{C}}^n|$ for a random constant composition codebook given in \eqref{eq:composition_codebook}.
	The trace distance between the induced state ${\rho}_{B^n}^{\breve{\mathcal{C}}^n}$ and the true marginal state $\breve{\rho}_{B^n}$ is upper bounded by
	\begin{align}
		\frac12\mathds{E}_{\breve{\mathcal{C}}^n} \left\|{\rho}_{B^n}^{\breve{\mathcal{C}}^n} - \breve{\rho}_{B^n} \right\|_1
		\leq
		\mathrm{e}^{  - n  \sup_{\alpha \in (1,2)} \frac{1-\alpha}{\alpha} \left( \breve{I}_\alpha^{*} \left( X{\,:\,}B \right)_\rho - R \right) }.
	\end{align}
	Here, the order-$\alpha$ sandwiched Augustin information $\breve{I}_\alpha^{*} ( X{\,:\,}B )_\rho$ is defined in \eqref{eq:sandwiched_Augustin}. Moreover, the exponent $\sup_{\alpha \in (1,2)} \frac{1-\alpha}{\alpha} ( \breve{I}_\alpha^{*} ( X{\,:\,}B )_\rho - R )$ is positive if and only if $R> I(X{\,:\,}B)_\rho$.
\end{theo}

\begin{remark}
	Jensen's inequality together with the concavity of logarithmic function show that
	\begin{align}
		\tfrac{1-\alpha}{\alpha} {I}_{\alpha}^{*}(X{\,:\,}B)_\rho  \leq \tfrac{1-\alpha}{\alpha} {I}_{\alpha}^{*}(X{\,:\,}B)_\rho, \quad \forall \alpha >1.
	\end{align}
	Hence, the expected value of the trace distance decays faster when using random constant composition codebook compared to that of using random i.i.d.~codebook.
\end{remark}

\begin{proof}The idea is similar to Theorem \ref{theo:achievability_iid_oneshot}.
	We write $\mathcal{H}\equiv \mathcal{H}_B$ throughout the proof.	
	For $x^n=x_1\cdots x_n\in \mathcal{X}^n$,
	we write the output state as $\rho_{x^n}:=\rho_{x_1}\otimes \cdots\otimes \rho_{x_n}\in \mathcal{B(H)}^{\otimes n}$.
	Let $x^n: \Omega\to \mathcal{X}^n$ be a random codeword with respect to the uniform distribution on type class $\breve{p}_{X^n}$.
	We introduce the classical-quantum state as
	\[\breve{\rho}_{\Omega B^n}=\sum_{x^n\in T_p^n} 1_{A_{x^n}}\ten \rho_{x^n} \in L_\infty\left(\Omega, \mathcal{B(H)}^{\otimes n}\right)\pl,\]
	where $1_{A_{x^n}}$ is the characteristic function on the mutually disjoint set $A_{x^n}$ such that
	$\Pr(A_{x^n})=\frac{1}{|T_p^n|}$. It is clear that $\mathbb{E}_{\Omega} \left[\breve{\rho}_{\Omega B^n}\right]=\breve{\rho}_{B^n}$.
	Take $1<\al< 2$ and $\frac{1}{\al}+\frac{1}{\al'}=1$.
	Let $n_x := n p_X(x)$ for some integer $n_x \in \mathbb{N}$ for all $x\in \mathcal{X}$.
	The Augustine information $\breve{I}_\alpha^{*} \left( X{\,:\,}B \right)_\rho$ can be expressed as
	\begin{align*}
		&\inf_{\sigma\in\mathcal{S}(\mathcal{H})} \frac{\alpha}{\alpha-1}  \sum_{x\in\mathcal{X}} \frac{n_x}{n} \log \left\| \sigma^{-\frac{1}{2\al'}}  \rho_x \sigma^{-\frac{1}{2\al'}}  \right\|_\alpha
		\\
		& \overset{}{=} \inf_{\sigma\in\mathcal{S}(\mathcal{H})} \frac{1}{n} \cdot \frac{\alpha}{\alpha-1} \log \left\| (\sigma^{\ten n})^{-\frac{1}{2\al'}}  \rho_{x^n} (\sigma^{\ten n})^{-\frac{1}{2\al'}} \right\|_\alpha^{n_x}
	\end{align*}
	for any $x_n\in T_p^n$. Then we further have
	\begin{align*}
		\breve{I}_\alpha^{*} \left( X{\,:\,}B \right)_\rho 
		&= \inf_{\sigma\in\mathcal{S}(\mathcal{H})} \frac{1}{n} \cdot \frac{\alpha}{\alpha-1} \log \norm{ \sigma_{\Omega B^n}^{-\frac{1}{2\al'}}\breve{\rho}_{\Omega B^n} \sigma_{\Omega B^n}^{-\frac{1}{2\al'}}}{L_\al(\Omega, S_\al(\mathcal{H}^{\otimes n}))}.
	\end{align*}
	where $\sigma_{\Omega B^n}=1_\Omega \ten \sigma_B^{\otimes n}$ is interpreted as a constant function with value $ \sigma_B^{\otimes n}\in \mathcal{B(H)}^{\otimes n}\cong \mathcal{B}(\mathcal{H}^{\otimes n})$. In other words,
	\begin{align*}\e^{ n\cdot \frac{\alpha-1}{\alpha } \breve{I}_\alpha^{*} \left(X{\,:\,}B \right)_\rho} =  \inf_{\sigma\in\mathcal{S}(\mathcal{H})} \norm{ \sigma_{\Omega B_n}^{-\frac{1}{2\al'}}\breve{\rho}_{\Omega B^n} \sigma_{\Omega B_n}^{-\frac{1}{2\al'}}}{L_\al(\Omega, S_\al(\mathcal{H}^{\otimes n}))}.
	\end{align*}
	Denote by $M:=|\breve{\mathcal{C}}^n|$.
	Now using the construction given in the Lemma~\ref{lemm:inter}, we have
	\begin{align*}\mathds{E}_{\breve{\mathcal{C}}^n} \left\|{\rho}_{B^n}^{\breve{\mathcal{C}}^n} - \breve{\rho}_{B^n} \right\|_1
		&= \norm{\Theta (\rho_{\Omega B^n})}{L_1(\Omega^{ M},S_1(\mathcal{H}^{\ten n}))}.
	\end{align*}
	Note that for any $ \sigma\in\mathcal{S}(\mathcal{H})$,
	\begin{align*}
		&\norm{\Theta (\rho_{\Omega B^n})}{L_1(\Omega^{ M},S_1(\mathcal{H}^{\ten n}))} \\
		&\overset{\text{(a)}}{\le} \;\norm{\sigma_{\Omega B^n}^{-\frac{1}{2\al'}}\breve{\rho}_{\Omega B^n} \sigma_{\Omega B^n}^{-\frac{1}{2\al'}}}{L_\al(\Omega^{ M},S_\al(\mathcal{H}^{\ten n}))}
		\\
		&\overset{\text{(b)}}{=} \;\norm{\Theta (\sigma_{\Omega B^n}^{-\frac{1}{2\al'} }\breve{\rho}_{\Omega B^n}\sigma_{\Omega B^n}^{-\frac{1}{2\al'} })}{L_\al(\Omega^{ M},S_\al(\mathcal{H}^{\ten n}))}
		\\
		&\le \;\norm{\Theta: L_\al(\Omega,S_\al(\mathcal{H}^{\ten n}))\to L_\al(\Omega^{ M},S_\al(\mathcal{H}^{\ten n}))}{} \cdot \norm{ \sigma_{\Omega B^n}^{-\frac{1}{2\al'}}\breve{\rho}_{\Omega B^n} \sigma_{\Omega B^n}^{-\frac{1}{2\al'}}}{L_\al(\Omega, S_\al(\mathcal{H}^{\otimes n}))}.
	\end{align*}
	Here, (b) follows from the fact $\Theta= \Theta \ten \id_{\mathcal{B(H)}}$ is identity on the operator part, and (a) uses H\"older inequality with
	\begin{align*}
		\norm{ \sigma_{B^n}^{\frac{1}{2\al'} }}{L_{2\al}(\Omega^{ M},S_{2\al}(\mathcal{H}^{\ten n}))}
		&=
		\Big(\int_{\Omega^{ M}} \norm{(\sigma^{\frac{1}{2\al'} })^{\ten n}}{S_{2\al'}(\mathcal{H}^{\ten n})}^{2\al'}d\mu^{ M}\Big)^{\frac{1}{2\al'}}\\
		&=1\pl.
	\end{align*}
	Then the assertion follows from Lemma~\ref{lemm:inter} and taking infimum over $\sigma\in\mathcal{S(H)}$, i.e.~
	\begin{align*}\mathds{E}_{\breve{\mathcal{C}}^n} \left\|{\rho}_{B^n}^{\breve{\mathcal{C}}^n} - \breve{\rho}_{B^n} \right\|_1 
		&\le  2^{\frac{2}{\alpha}-1}M^{\frac{1-\alpha}{\alpha}}  \inf_{\sigma\in\mathcal{S}(\mathcal{H})}  \norm{ \sigma_{\Omega B_n}^{-\frac{1}{2\al'}}\breve{\rho}_{\Omega B^n} \sigma_{\Omega B_n}^{-\frac{1}{2\al'}}}{L_\al(\Omega, S_\al(\mathcal{H}^{\otimes n}))}\\ 
		&= 2^{\frac{2}{\alpha}-1}M^{\frac{1-\alpha}{\alpha}} \e^{ n\cdot \frac{\alpha-1}{\alpha } \breve{I}_\alpha^{*} \left(X{\,:\,}B \right)_\rho}\\
		&\le 2 \e^{ -n\cdot \frac{1-\alpha}{\alpha } (\breve{I}_\alpha^{*} \left(X{\,:\,}B \right)_\rho-R)},
	\end{align*}
	where $R= \frac1n \log |{\mathcal{C}}^n|$. The assertion of exponential decay follows from taking infimum of the right-hand side for $\alpha\in (1,2)$.
	The positivity again follows from the monotone non-decreasing of the map $\alpha\mapsto \breve{I}_\alpha^{*}$ \cite[Proposition~5]{CGH18} and \eqref{eq:alpha1}.
\end{proof}

\section{Strong Converse} \label{sec:sc}
In the previous Section~\ref{sec:achievability}, we have presented that as long as the rate of the random codebook size is \emph{above} the quantum mutual information $I(X{\,:\,}B)_\rho$, the trace distances using both the random i.i.d.~codebook and the random constant composition codebook exponentially decay.
In this section, we show that, on the other hand, when the rate of the random codebook size is \emph{below} the quantum mutual information $I(X{\,:\,}B)_\rho$, the trace distances using both the two random codebooks converge to $1$ exponentially fast, reflecting the \emph{exponential strong converse }.

Using the notation as in Section~\ref{sec:achievability}, we first prove the following one-shot strong converse bound. 
\begin{theo}[A one-shot strong converse] \label{theo:sc}
	The trace distance between the induced  state ${\rho_{B}^\mathcal{C}}$ and the true state $\rho_B$ is lower bounded by, 
	\begin{align}
		\frac12\mathds{E}_\mathcal{C} \left\|\rho_{B}^\mathcal{C} - \rho_B \right\|_1
		\geq 1 - 4\, \mathrm{e}^{\frac{\alpha-1}{\alpha} \left(  I_{2-\sfrac{1}{\alpha}}^{\downarrow}(X{\,:\,}B)_\rho - \log M \right) }
				, \quad \forall \alpha \in (\sfrac12, 1).
	\end{align}
	Here, $I^\downarrow_{2-\sfrac{1}{\alpha}}(X{\,:\,}B)_\rho$ is defined in \eqref{eq:Petz_Renyi_down}.
\end{theo}
In Sections~\ref{sec:sc_iid} and \ref{sec:sc_cc} later, we will apply the one-shot strong converse, Theorem~\ref{theo:sc}, to the random i.i.d.~codebook and the random constant composition codebook.

\begin{proof}
	
	Using the Holevo--Helstrom theorem \cite{Hel67, Hol72}, i.e.~
	\begin{align}
		\frac12\|\rho-\sigma\|_1 = 
	\sup_{0\leq \Pi\leq \mathds{1}}  \Tr\left[(\rho-\sigma)\Pi \right],
	\end{align}
	we have
	\begin{align}
		\frac12 \left\|\rho_{B}^\mathcal{C} - \rho_B \right\|_1
		&=\sup_{0\leq \Pi_B \leq \mathds{1}_B} \Tr\left[ \left(\rho_{B}^\mathcal{C} - \rho_B\right) \Pi_B \right] \\
		\begin{split} \label{eq:sc1}
			&\geq
			\Tr\left[ \rho_{B}^\mathcal{C} \left( \rho_{B}^\mathcal{C} + \rho_B \right)^{-\sfrac12} \rho_{B}^\mathcal{C} \left( \rho_{B}^\mathcal{C} + \rho_B \right)^{-\sfrac12} \right] \\
			&\quad- \Tr\left[ \rho_{B} \left( \rho_{B}^\mathcal{C} + \rho_B \right)^{-\sfrac12} \rho_{B}^\mathcal{C} \left( \rho_{B}^\mathcal{C} + \rho_B \right)^{-\sfrac12} \right].
		\end{split}
	\end{align}
	We will then lower bound the two terms in \eqref{eq:sc1} subsequently.
	
	Recalling ${\rho}_{B}^\mathcal{C} = \frac1M \sum_{x\in\mathcal{C}} \rho_B^x$, we rewrite the first term in \eqref{eq:sc1} as follows:
	\begin{align}
		&\Tr\left[ \rho_{B}^\mathcal{C} \left( \rho_{B}^\mathcal{C} + \rho_B \right)^{-\sfrac12} \rho_{B}^\mathcal{C} \left( \rho_{B}^\mathcal{C} + \rho_B \right)^{-\sfrac12} \right] \notag\\
		&= \frac1M \sum_{x\in\mathcal{C}} \Tr\left[ \rho_{B}^{x} \left( \rho_{B}^\mathcal{C} + \rho_B \right)^{-\sfrac12} \rho_{B}^\mathcal{C} \left( \rho_{B}^\mathcal{C} + \rho_B \right)^{-\sfrac12} \right] \\
		&= \frac1M \sum_{x\in\mathcal{C}} \Tr\left[ \rho_{B}^{x} \left( \sum_{\bar x\in\mathcal{C}} \rho_B^{\bar x} + M\rho_B \right)^{-\sfrac12} \left( \rho_B^x +\sum_{\bar x\in\mathcal{C} , \bar x\neq x} \rho_B^{\bar x} \right) \left( \sum_{\bar x\in\mathcal{C}} \rho_B^{\bar x} + M\rho_B \right)^{-\sfrac12} \right] \\
		&\geq \frac1M \sum_{x\in\mathcal{C}} \Tr\left[ \rho_{B}^{x} \left( \sum_{\bar x\in\mathcal{C}} \rho_B^{\bar x} + M\rho_B \right)^{-\sfrac12}  \rho_B^x  \left( \sum_{\bar x\in\mathcal{C}} \rho_B^{\bar x} + M\rho_B \right)^{-\sfrac12} \right] \\
		&\geq \frac1M \sum_{x\in\mathcal{C}} \left( 1 - \Tr\left[ (\rho_B^x)^{1-s} \left( \sum_{\bar x\in\mathcal{C} , \bar x\neq x} \rho_B^{\bar x} + M \rho_B \right)^s \right] \right), \quad \forall  s\in(0,1),
	\end{align}
	where we have applied Lemma~\ref{lemm:trace} given below with $K = \rho_B^x$ and $L =  \sum_{{\bar  x\in\mathcal{C} , \bar x\neq x}}\rho_B^{\bar x} + M \rho_B$ to the last inequality.
	
	Recalling the fact that each codeword (e.g.~$x,\bar x\in\mathcal{C}$) is drawn independently, and using the operator concavity of $(\,\cdot\,)^s$ for $s\in(0,1)$, we average the above inequality over the random codebook $\mathcal{C}$ to arrive at
	\begin{align}
		&\mathds{E}_\mathcal{C} \Tr\left[ \rho_{B}^\mathcal{C} \left( \rho_{B}^\mathcal{C} + \rho_B \right)^{-\sfrac12} \rho_{B}^\mathcal{C} \left( \rho_{B}^\mathcal{C} + \rho_B \right)^{-\sfrac12} \right] \notag\\
		&\geq \frac1M \sum_{x\in\mathcal{C}} \left( 1 - \mathds{E}_{x\sim p_X } \Tr\left[ (\rho_B^x)^{1-s} \mathds{E}_{\bar x \sim p_X}\left( \sum_{{\bar x \in\mathcal{C} , \bar x\neq x}} \rho_B^{\bar x} + M \rho_B \right)^s \right] \right) \\
		&\geq \frac1M \sum_{x\in\mathcal{C}} \left( 1 - \mathds{E}_{x\sim p_X } \Tr\left[ (\rho_B^x)^{1-s} \left( \sum_{{\bar x\in\mathcal{C} , \bar x\neq x}} \mathds{E}_{\bar x \sim p_X}\rho_B^{x'} + M \rho_B \right)^s \right] \right) \\
		&= \frac1M \sum_{x\in\mathcal{C}} \left( 1 - \mathds{E}_{x\sim p_X } \Tr\left[ (\rho_B^x)^{1-s} \left( \sum_{{x'\in\mathcal{C} , \bar x\neq x}} (M-1)\rho_B + M \rho_B \right)^s \right] \right) \\
		&\geq 1-(2M)^s \mathds{E}_{x\sim p_X}\Tr\left[ (\rho_B^x)^{1-s} \rho_B ^s \right], \quad \forall  s\in(0,1). \label{eq:sc2}
	\end{align}
	
	Next, we lower bound the second term in \eqref{eq:sc1}. Using the cyclic property of trace, we have
	\begin{align}
		&-\Tr\left[ \rho_{B} \left( \rho_{B}^\mathcal{C} + \rho_B \right)^{-\sfrac12} \rho_{B}^\mathcal{C} \left( \rho_{B}^\mathcal{C} + \rho_B \right)^{-\sfrac12} \right] \notag \\
		&= -\frac1M \sum_{x\in\mathcal{C}} \Tr\left[\rho_B^x \left( \rho_{B}^\mathcal{C} + \rho_B \right)^{-\sfrac12} \rho_{B} \left( \rho_{B}^\mathcal{C} + \rho_B \right)^{-\sfrac12} \right] \\
		&= -\frac1M \sum_{x\in\mathcal{C}} \Tr\left[ \rho_{B}^{x} \left( \sum_{\bar x\in\mathcal{C}} \rho_B^{\bar x} + M\rho_B \right)^{-\sfrac12} M\rho_B \left( \sum_{\bar x\in\mathcal{C}} \rho_B^{\bar x} + M\rho_B \right)^{-\sfrac12} \right] \\
		&\geq -\frac1M \sum_{x\in\mathcal{C}} \Tr\left[ \rho_{B}^{x} \left( \sum_{\bar x\in\mathcal{C}} \rho_B^{\bar x} + M\rho_B \right)^{-\sfrac12} \left( \sum_{{\bar x\in\mathcal{C} , \bar  x\neq x}} \rho_B^{\bar x} + M \rho_B\right) \left( \sum_{x\in\mathcal{C}} \rho_B^x + M\rho_B \right)^{-\sfrac12} \right] \\
		&\geq -\frac1M \sum_{x\in\mathcal{C}} \Tr\left[ \left(\rho_{B}^{x}\right)^{1-s} \left( \sum_{{\bar x\in\mathcal{C} , \bar x\neq x}} \rho_B^{\bar x} + M \rho_B \right)^s \right], \quad \forall s \in (0,1),
	\end{align}
	where we invoked Lemma~\ref{lemm:trace} again with $K = \rho_B^x$ and $L =  \sum_{{\bar x\in\mathcal{C} , \bar x\neq x}}\rho_B^{\bar x} + M \rho_B$ to the last inequality.
	
	Similar, we take averaging over the random codebook $\mathcal{C}$ and follow previous reasoning to have
	\begin{align}
		- \mathds{E}_{\mathcal{C}}\Tr\left[ \rho_{B} \left( \rho_{B}^\mathcal{C} + \rho_B \right)^{-\sfrac12} \rho_{B}^\mathcal{C} \left( \rho_{B}^\mathcal{C} + \rho_B \right)^{-\sfrac12} \right]
		&\geq -(2M)^s \mathds{E}_{x\sim p_X}\Tr\left[ (\rho_B^x)^{1-s} \rho_B ^s \right], \quad \forall s\in(0,1). \label{eq:sc3}
	\end{align}
	
	Combining \eqref{eq:sc1}, \eqref{eq:sc2}, and \eqref{eq:sc3} proves our claim with substitution $\alpha = \frac{1}{1+s}$.
\end{proof}

\begin{lemm}[A trace inequality {\cite[Lemma 3]{SGC22a}}] \label{lemm:trace}
	For any positive semi-definite $K$ and $L$, the following holds,
	\begin{align}
		\Tr\left[ K (K+L)^{-\sfrac12} L (K+L)^{-\sfrac12} \right] \leq \Tr\left[ K^{1-s} L^s \right], \quad \forall s\in (0,1).
	\end{align}
\end{lemm}

\subsection{Random I.I.D.~Codebook} \label{sec:sc_iid}
In the following, we consider the $n$-shot scenario of quantum soft covering with i.i.d.~prior $p_X^{\otimes n}$ as stated in Section~\ref{sec:achievability_iid}.
We then apply the one-shot strong converse, Theorem~\ref{theo:sc} with the random i.i.d.~codebook. We show that the quantum  mutual information $I(X{\,:\,}B)_\rho$ is the \emph{strong converse rate} of the quantum soft covering, meaning that the trace distance exponentially converges to $1$ when the rate of the random codebook size is below $I(X{\,:\,}B)_\rho$.

\begin{prop}[Exponential strong converse using random i.i.d.~codebook] \label{prop:sc_iid}
	Let $\rho_{XB} = \sum_{x\in\mathcal{X}} p_X(x)|x\rangle \langle x|\otimes \rho_B^x$ be a classical-quantum state, and let $R>0$. For any $n\in\mathds{N}$, let
	$\mathcal{C}^n:=\{x^n(1),\ldots, x^n(M)\}$ be a random i.i.d.~codebook with rate $R = \frac1n \log |\mathcal{C}^n|$, where each $x^n(m)$ is independently drawn from $p_X^{\otimes n}$.	Then, for any $n\in\mathds{N}$,
	\begin{align}
		\frac12\mathds{E}_{\mathcal{C}^n}  \left\|\rho_{B^n}^{\mathcal{C}^n}  -  \rho_B^{\otimes n} \right\|_1
		& \geq 1  -  4\,\mathrm{e}^{-n\sup\limits_{\alpha \in (\sfrac12,1)} \frac{1-\alpha}{\alpha}\left( I_{2-\frac{1}{\alpha}}^{\downarrow}(X{\,:\,}B)_\rho - R\right) }.
	\end{align}
	\noindent Moreover, the exponent $\sup_{\alpha \in (\sfrac12,1)} \frac{1-\alpha}{\alpha}( I_{2-\sfrac{1}{\alpha}}^{\downarrow}(X{\,:\,}B)_\rho - R)$ is positive if and only if $R< I(X{\,:\,}B)_\rho$.
\end{prop}
\begin{proof}The estimate follows from
	Theorem~\ref{theo:sc} and the additivity $I_{2-\sfrac{1}{\alpha}}^{\downarrow}(X^n;B^n)_{\rho^{\otimes n}} = n I_{2-\sfrac{1}{\alpha}}^{\downarrow}(X{\,:\,}B)_\rho$.
	The positivity follows from the monotonicity of $\alpha\mapsto D_\alpha$ \cite[Lemma 3.12]{MO14} and \eqref{eq:alpha1}.
\end{proof}

\subsection{Random Constant Composition Codebook} \label{sec:sc_cc}
\begin{prop}[Exponential strong converse using random constant composition codebook] \label{prop:sc_cc}
	For any $n\in\mathds{N}$, consider a classical-quantum state $\rho_{XB} = \sum_{x\in\mathcal{X}} p_X(x) |x\rangle\langle x|\otimes \rho_B^x$, where $n p_X(x) \in 0\cup\mathbb{N}$ for all $x\in\mathcal{X}$, and let $R:= \frac1n \log |\breve{\mathcal{C}}^n|$, where $\breve{\mathcal{C}}^n$ for a random constant composition codebook given in \eqref{eq:composition_codebook}.
	Then, there exists $k_p>0$ only depending on $p_X$ such that
	\begin{align}
	\frac12\mathds{E}_{\breve{\mathcal{C}}^n} \left\|\rho_{B^n}^{\breve{\mathcal{C}}^n} - \breve{\rho}_{B^n} \right\|_1 \geq 1 - 
	n^{k_p} \e^{ -
	n\sup_{\alpha \in (\sfrac12,1)} \frac{1-\alpha}{\alpha}\left( \breve{I}_{2-\sfrac{1}{\alpha}}^{\downarrow}(X{\,:\,}B)_\rho - R\right) }.
	\end{align}
	Moreover, the exponent $\sup_{\alpha \in (\sfrac12,1)} \frac{1-\alpha}{\alpha}( \breve{I}_{2-\sfrac{1}{\alpha}}^{\downarrow}(X{\,:\,}B)_\rho - R)$ is positive if and only if $R< I(X{\,:\,}B)_\rho$.
\end{prop}

\begin{remark}
	Jensen's inequality together with the concavity of logarithmic function show that
	\begin{align}
		\tfrac{1-\alpha}{\alpha}{I}_{\alpha}^{\downarrow}(X{\,:\,}B)_\rho  \leq \tfrac{1-\alpha}{\alpha}\breve{I}_{\alpha}^{\downarrow}(X{\,:\,}B)_\rho, \quad \forall \alpha \in (0,1).
	\end{align}
	Hence, Propositions~\ref{prop:sc_iid} and \ref{prop:sc_cc} show that the expected value of the trace distance using random composition codebook converges to $1$ faster than that of using random i.i.d.~codebook, albeit with a vanishing higher-order term.
\end{remark}

\begin{proof}
	Applying the one-shot strong converse established in Theorem~\ref{theo:sc} with prior distribution $\breve{p}_{X^n}$ and the random constant composition codebook $\breve{\mathcal{C}}^n$, we obtain
	\begin{align} \label{eq:sc_cc1}
		\frac12\mathds{E}_{\breve{\mathcal{C}}^n} \left\|\rho_{B^n}^{\breve{\mathcal{C}}^n} - \breve{\rho}_{B^n} \right\|_1
		\geq 1 - 4 M^s \mathds{E}_{x^n\sim \breve{p}_{X^n}} \Tr\left[ \left( \rho_{B^n}^{x^n} \right)^{1-s} \left(\breve{\rho}_{B^n}\right)^s \right], \quad \forall s\in(0,1), \; n\in\mathds{N}.
	\end{align}
	Note that
	\begin{align} \label{eq:sc_cc2}
		\breve{\rho}_{B^n} = \sum_{x^n\in\mathcal{X}^n} \frac{ \mathbf{1}_{x^n \in T_{p}^n} p_X^{\otimes n}(x^n) \rho_{B^n}^{x^n} }{p_X^{\otimes n}\left( T_{p}^n \right)  }
		\leq \sum_{x^n\in\mathcal{X}^n} \frac{ p_X^{\otimes n}(x^n) \rho_{B^n}^{x^n} }{p_X^{\otimes n}\left( T_{p}^n \right)  }
		&= \frac{\rho_B^{\otimes n} }{ p_X^{\otimes n}\left( T_{P}^n \right)  }.
	\end{align}
	Since $(\,\cdot\,)^s$ is operator monotone for $s\in(0,1)$, we combine \eqref{eq:sc_cc1} and \eqref{eq:sc_cc2} to get
	\begin{align}
		\frac12\mathds{E}_{\breve{\mathcal{C}}^n} \left\|\rho_{B^n}^{\breve{\mathcal{C}}^n} - \breve{\rho}_{B^n} \right\|_1
		&\geq 1 - 4 M^s \mathds{E}_{x^n\sim \breve{p}_{X^n}} \Tr\left[ \left( \rho_{B^n}^{x^n} \right)^{1-s} \left(\rho_B^{\otimes n}\right)^s \right] \left(p_X^{\otimes n}\left( T_{p}^n \right)\right)^{-s}\\
		&= 1 - 4 M^{\frac{1-\alpha}{\alpha}} \mathrm{e}^{n \frac{\alpha-1}{\alpha} \breve{I}_{2-\sfrac{1}{\alpha}}^{\downarrow}(X{\,:\,}B)_\rho }  \left(p_X^{\otimes n}\left( T_{p}^n \right)\right)^{\frac{\alpha-1}{\alpha}},
		\quad \forall \alpha\in(\sfrac12,1),
	\end{align}
	where we have used substitution $\alpha = \frac{1}{1+s}$.
	
	By \cite[p.~26]{CK11}, the probability of the set of all sequences with composition $P$ under $P^{\otimes n}$ is
	\begin{align}
		p_X^{\otimes n}\left( T_{p}^n \right) = \mathrm{e}^{- \xi \frac{|\texttt{supp}(p_X)|}{12 \log 2}} ( 2\pi n)^{- \frac{|\texttt{supp}(p_X)|-1}{2} } \sqrt{ \prod_{x:p_X(x)>0} \frac{1}{p_X(x)} }
	\end{align}
	for some $\xi\in[0,1]$.
	
	Taking
	\begin{align}
		K_p := \frac{|\texttt{supp}(p_X)|}{12 \log 2} + \frac{|\texttt{supp}(p_X)|-1}{2}\cdot\log (2\pi) + \frac12\sum_{x\in\texttt{supp}(p_X)} \log p_X(x)  + \log 4,
	\end{align}
	thus proves our claim of exponential decay.
	
	Again, the positivity follows from the non-decreasing map $\alpha\mapsto D_\alpha$ \cite[Lemma 3.12]{MO14} and \eqref{eq:alpha1}.
\end{proof}

\section{Moderate Deviation Analysis} \label{sec:moderate}

In previous sections, we study the large deviation analysis for quantum soft covering. We characterize the exponential error behaviors when the random codebook size is fixed.
In this section, we extends our results to the moderate deviation regime \cite{CH17, CTT2017}. That is, we derive the asymptotic error behaviors (in terms of trace distance) when the rate $R_n$ of codebook size (as a function of blocklength $n$) approaches $I(X{\,:\,}B)_{\rho}$ at certain speed.
The central question we want to ask here is that if  $R_n$ approaches $I(X{\,:\,}B)_{\rho}$ only \emph{moderately quickly}, can the trace distance still vanish?
As will be shown in the following Propositions~\ref{prop:moderate_iid} and \ref{prop:moderate_cc}, the answers are affirmative for both the random i.i.d.~codebook and the random constant composition codebook when  $R_n$ approaches $I(X{\,:\,}B)_{\rho}$ no faster than $O(\sfrac{1}{\sqrt{n}})$.

We call $(a_n)_{n\in\mathds{N}}$ a moderate deviation sequence if it satisfies
\begin{align} \label{eq:an}
	\lim_{n\to\infty} a_n = 0, \quad
	\lim_{n\to\infty} n a_n^2 = \infty.
\end{align}

\begin{prop}[Moderate deviations using random i.i.d.~codebook] \label{prop:moderate_iid}
	Let $\rho_{XB}$ be a classical-quantum state satisfying $V(X{\,:\,}B)_{\rho}>0$.
	We have the following result for any moderate deviation sequence $(a_n)_{n\in\mathds{N}}$:
	\begin{align}
		\begin{dcases}
			\liminf_{n\to \infty} - \frac{1}{n a_n^2} \log \left( \frac12 \mathbb{E}_{\mathcal{C}^n}\left\| \rho_{B^n}^{\mathcal{C}^n} - \rho_B^{\otimes n}\right\|_1 \right) \geq \frac{1}{2 V(X{\,:\,}B)_\rho},
			& \text{ if } |\mathcal{C}^n| = \e^{n(I(X{\,:\,}B)_\rho + a_n)} \\
			\liminf_{n\to \infty} - \frac{1}{n a_n^2} \log \left( 1-\frac12 \mathbb{E}_{\mathcal{C}^n}\left\| \rho_{B^n}^{\mathcal{C}^n} - \rho_B^{\otimes n}\right\|_1 \right) \geq \frac{1}{2 V(X{\,:\,}B)_\rho}, & \text{ if } |\mathcal{C}^n| = \e^{n(I(X{\,:\,}B)_\rho - a_n)}
		\end{dcases}.
	\end{align}
\end{prop}

\begin{prop}[Moderate deviations using constant composition random codebooks] \label{prop:moderate_cc}
	For any $n\in\mathbb{N}$, consider a c-q state $\rho_{XB} = \sum_{x\in\mathcal{X}} p_X(x) |x\rangle\langle x|\otimes \rho_B^x$, where $n p_X(x) \in 0\cup\mathbb{N}$ for all $x\in\mathcal{X}$, and consider a random constant composition codebook given in \eqref{eq:composition_codebook}.
	Suppose $\breve{V}(X{\,:\,}B)_{\rho}>0$.
	We have the following result for any moderate deviation sequence $(a_n)_{n\in\mathds{N}}$ defined in \eqref{eq:an}:
	\begin{align}
		\liminf_{n\to \infty} - \frac{1}{n a_n^2} \log \left( \frac12\mathds{E}_{\breve{\mathcal{C}}^n} \left\|{\rho}_{B^n}^{\breve{\mathcal{C}}^n} - \breve{\rho}_{B^n} \right\|_1 \right) \geq \frac{1}{2 \breve{V}(X{\,:\,}B)_\rho},
		& \text{ if } |\breve{\mathcal{C}}^n| = \e^{n(I(X{\,:\,}B)_\rho + a_n)}.
	\end{align}
\end{prop}

\begin{remark} \label{remark:moderate_cc}
	We note that for rate below $I(X{\,:\,}B)_\rho$, the following statement
	\begin{align} \label{eq:moderate_cc1}
		\liminf_{n\to \infty} - \frac{1}{n a_n^2} \log \left( 1 - \frac12\mathds{E}_{\breve{\mathcal{C}}^n} \left\|{\rho}_{B^n}^{\breve{\mathcal{C}}^n} - \breve{\rho}_{B^n} \right\|_1 \right) \geq \frac{1}{2 \breve{V}(X{\,:\,}B)_\rho }, & \text{ if } |\breve{\mathcal{C}}^n| = \e^{n(I(X{\,:\,}B)_\rho - a_n)}
	\end{align}
	holds for a kind of moderate deviation sequence $a_n = \Theta(n^{-t})$ for any $t \in (0,\sfrac12)$ which is a special case of \eqref{eq:an}.
	Nonetheless, \eqref{eq:moderate_cc1} studies the situation where the trace converges to $1$, which is of less practical importance than the characterization of vanishing error scenario given in Proposition~\ref{prop:moderate_cc}.
\end{remark}

Before proving our claims, we shall employ the following first-order derivatives of those entropic information quantities appearing in the exponent functions.

\begin{lemm}[{\cite[Proposition 11]{Tomhay16}, \cite{CH17}}]\label{lemm:updiff}
	For every classical-quantum state $\rho_{XB}$, the map $\alpha \mapsto I^*_{\alpha}(X{\,:\,}B)_{\rho}$ are continuously differentiable on $\alpha \in [1,2]$, and maps
	$\alpha \mapsto I^{\downarrow}_{2-\sfrac{1}{\alpha}}(X{\,:\,}B)_{\rho}$ and $\alpha \mapsto \breve{I}^{\downarrow}_{2-\sfrac{1}{\alpha}}(X{\,:\,}B)_{\rho}$ are analytical on $\alpha \in [\sfrac12,1]$.
	Moreover,
	\begin{align}
		\left.\frac{\mathrm{d}}{\mathrm{d}\alpha}I^*_{\alpha}(X{\,:\,}B)_{\rho}\right|_{\alpha = 1} = \left.\frac{\mathrm{d}}{\mathrm{d}\alpha}I^{\downarrow}_{2-\sfrac{1}{\alpha}}(X{\,:\,}B)_{\rho}\right|_{\alpha = 1} = \frac{V(X{\,:\,}B)_{\rho}}{2}, \quad
		\left.\frac{\mathrm{d}}{\mathrm{d}\alpha}\breve{I}^{\downarrow}_{2-\sfrac{1}{\alpha}}(X{\,:\,}B)_{\rho}\right|_{\alpha = 1} = \frac{\breve{V}(X{\,:\,}B)_{\rho}}{2}.	
	\end{align}
\end{lemm}

\begin{lemm}\label{lemm:Augustin}
	For every classical-quantum state $\rho_{XB}$, the map $\alpha \mapsto \breve{I}^*_{\alpha}(X{\,:\,}B)_{\rho}$ is continuously differentiable on $\alpha \in [1,2]$.
	Moreover,
	\begin{align}
		\left.\frac{\mathrm{d}}{\mathrm{d}\alpha}\breve{I}^*_{\alpha}(X{\,:\,}B)_{\rho}\right|_{\alpha = 1} = \frac{\breve{V}(X{\,:\,}B)_{\rho}}{2}.	
	\end{align}
\end{lemm}
\begin{proof}[Proof of Lemma~\ref{lemm:Augustin}]
	We adopt the short notation $p\equiv p_X$
	$\sigma=\sigma_B$ and $\rho_x=\rho_B^x$.
	Let  $\sigma_{\alpha, p}^\star$ be the \emph{order-$\alpha$ Augustin mean} that attains the infimum in the definition of the order-$\alpha$ Augustin information, i.e.~
	\begin{align}
		\breve{I}^*_{\alpha}(X{\,:\,}B)_\rho = \mathds{E}_{x\sim p} D_\alpha^*\left(\rho_x \| \sigma_{\alpha,p}^\star \right).
	\end{align}
	From definition of the sandwiched Augustin information $\breve{I}^*_{\alpha}$ given in \eqref{eq:sandwiched_Augustin}, we know that the map
	\begin{align}
		(\alpha,\sigma_B) \mapsto \mathds{E}_{x\sim p} D_\alpha^*\left(\rho_x \| \sigma_B \right)
	\end{align}
	is twice Fr\'echet differentiable and the Augustin mean $\sigma_{\alpha, p}^\star$ exists \cite{CGH18}. It was shown in \cite[
Lemma 23]{HT14} that for any state $\rho$ and $\alpha>1$, the function
\[ \sigma \mapsto \left\|\sigma^{\frac{1-\alpha}{2\alpha}}\rho \sigma^{\frac{1-\alpha}{2\alpha}}\right\|_\alpha^\al\]
has strictly positive definite Hessian. Since $t\mapsto \log t$ is a strictly increasing function,
 the function
\begin{align}
		(\alpha,\sigma_B) \mapsto \mathds{E}_{x\sim p} D_\alpha^*\left(\rho_x \| \sigma_B \right)=\frac{1}{\alpha-1}  \sum_{x\in\mathcal{X}} p(x) \log \left\| \sigma_B^{\frac{1-\alpha}{2\alpha} }  \rho_x \sigma_B^{\frac{1-\alpha}{2\alpha}} \right\|_\alpha^\alpha.
	\end{align}
 has strictly positive Hessian with respect to $\sigma_B$ for all $\alpha\in (1,\infty)$. This is also true for $\alpha=1$.  Note that for the infimum it sufficient to consider $\sigma$ with $\text{supp}(\sigma)=\text{supp}(\rho_B)$. At $\alpha=1$
 \[ \sigma \mapsto \mathds{E}_{x\sim p} D\left(\rho_x \| \sigma \right)=\sum_{x}p(x)\tr\left[\rho_x\log \rho_x-\rho_x\log \sigma\right]\pl.\]
 Given a traceless Hermitian matrix $h$, we denote $\sigma_t=\sigma+th$. The Hessian is
 \begin{align*}\frac{\d^2 \Big(\mathds{E}_{x\sim p}D(\rho_x\|\sigma_t)\Big)}{\d t^2}=&\sum_{x}p(x) \int_{0}^\infty \tr\left[\rho_x(\sigma_t+s)^{-1}h(\sigma_t+s)^{-1}h(\sigma_t+s)^{-1} \right] \d s 
 \\=&  \int_{0}^\infty \tr\left[\rho_B(\sigma_t+s)^{-1}h(\sigma_t+s)^{-1}h(\sigma_t+s)^{-1} \right] \d s \, .
 \end{align*}
 Suppose $\rho_B\ge \mu \texttt{supp}(\rho_B)$ and $\sigma\le \lambda \texttt{supp}(\rho_B)$ for some $\mu,\lambda>0$. Then at $t=0$,
 \begin{align*}\left.\frac{\d^2 \sum_{x}p(x)D(\rho_x\|\sigma_t)}{\d t^2}\right|_{t=0}
 &=  \int_{0}^\infty \tr\left[ \rho_B(\sigma+s)^{-1}h(\sigma+s)^{-1}h(\sigma+s)^{-1} \right] \d s 
 \\ &\ge   \mu\int_{0}^\infty \tr\left[(\sigma+s)^{-2}h(\sigma+s)^{-1} h \right] \d s 
 \\ &\ge   \mu\int_{0}^\infty \tr\left[(\lambda+s)^{-2}h(\sigma+s)^{-1} h \right] \d s 
 \\ &\ge   \mu\int_{0}^\infty \tr\left[(\lambda+s)^{-3}h^2 \right]) \d s \\ 
 &= \frac{\mu}{2\lambda^2}\norm{h}{2}^2,
 \end{align*}
 which implies strictly positive Hessian. By \cite[Lemma 24]{HT14}, we have the continuous differentiability of the map $\alpha \mapsto \breve{I}^*_{\alpha}(X{\,:\,}B)_{\rho}$ on $[1,2]$ and that	
	\begin{align}
		\frac{\mathrm{d}}{\mathrm{d}\alpha}\breve{I}^*_{\alpha}(X{\,:\,}B)_{\rho} =
		\left. \frac{\partial}{\partial\alpha}\mathds{E}_{x\sim p} D_\alpha^*\left(\rho_x \| \sigma \right)\right|_{\sigma = \sigma_{\alpha,p}^\star }.
	\end{align}
	Using the fact  \cite{LT15}  that $\left.\frac{\mathrm{d}}{\mathrm{d}\alpha} D_\alpha(\rho||\sigma)\right|_{\alpha=1} = \frac12 V(\rho||\sigma)$ for fixed $\rho$ and $\sigma$
	completes our proof.
\end{proof}
Now, we are ready to prove our claims of moderate deviation analysis.
\begin{proof}[Proofs of Propositions~\ref{prop:moderate_iid} and \ref{prop:moderate_cc}]

We prove Proposition~\ref{prop:moderate_cc}. The proof of Proposition~\ref{prop:moderate_iid} follows from similar reasoning. For the first claim of Proposition~\ref{prop:moderate_cc}, by Theorem~\ref{theo:achievability_composition},
\begin{align}
		\frac12\mathds{E}_{\breve{\mathcal{C}}^n} \left\|{\rho}_{B^n}^{\breve{\mathcal{C}}^n} - \breve{\rho}_{B^n} \right\|_1
		\leq
		\mathrm{e}^{  - n  \sup_{\alpha \in (1,2)} \frac{1-\alpha}{\alpha} \left( \breve{I}_\alpha^{*} \left( X{\,:\,}B \right)_\rho - R_n \right) }.
\end{align}
where $R_n=\frac{1}{n}\log |\breve{\mathcal{C}}^n|$ is the rate of the random constant composition codebook size.
Using Lemma \ref{lemm:Augustin}, we apply Taylor's series expansion of $\mapsto \breve{I}^*_{\alpha}(X{\,:\,}B)_{\rho}$ at $\alpha=1$:
\begin{align}
	\breve{I}^*_{\alpha}(X{\,:\,}B)_{\rho} = I(X{\,:\,}B)_{\rho} + \frac{\alpha-1}{2} \breve{V}(X{\,:\,}B)_{\rho} + \mathscr{R}(\alpha -1),
\end{align}
where $\mathscr{R}(\alpha-1)$ is a continuous function satisfying $\frac{\mathscr{R}(\alpha-1)}{\alpha-1}\to 0$ as $\alpha\to 1$. Let $\alpha_n =  1+\frac{a_n}{\breve{V}(X{\,:\,}B)_{\rho}}$. Using the above expansion and $R_n = I(X{\,:\,}B)_{\rho} + a_n$, we have $1<\alpha_n\leq2$ for all sufficiently large $n\in\mathds{N}$, and
\begin{align}
	\sup\limits_{1<\alpha\leq2}\left\{\frac{1-\alpha}{\alpha}(\breve{I}^{*}_{\alpha}(X{\,:\,}B)_{\rho} - R_n)\right\}
	&\geq \frac{1-\alpha_n}{\alpha_n}(\breve{I}^{*}_{\alpha_n}(X{\,:\,}B)_{\rho} - R_n)\\
	&= \frac{1}{1+\frac{a_n}{\breve{V}(X{\,:\,}B)_{\rho}}}\left(\frac{a_n^2}{2V(X{\,:\,}B)_{\rho}}-\frac{a_n^2}{\breve{V}(X{\,:\,}B)_{\rho}^2}\frac{\mathscr{R}(\alpha_n-1)}{\alpha_n-1}\right)\\
	& =\frac{a_n^2}{2\breve{V}(X{\,:\,}B)_{\rho}}\frac{1}{1+\frac{a_n}{\breve{V}(X{\,:\,}B)_{\rho}}}\left(1-\frac{2}{\breve{V}(X{\,:\,}B)_{\rho}}\frac{\mathscr{R}(\alpha_n-1)}{\alpha_n-1}\right).
\end{align}
Hence,
\begin{align}
	- \frac{1}{n a_n^2} \log \left(\frac12\mathds{E}_{\breve{\mathcal{C}}^n} \left\|{\rho}_{B^n}^{\breve{\mathcal{C}}^n} - \breve{\rho}_{B^n} \right\|_1 \right) \geq -\frac{\log 2}{na_n^2}+\frac{1}{2\breve{V}(X{\,:\,}B)_{\rho}}\frac{1}{1+\frac{a_n}{\breve{V}(X{\,:\,}V)_{\rho}}}\left(1-\frac{2}{\breve{V}(X{\,:\,}B)_{\rho}}\frac{\mathscr{R}(\alpha_n-1)}{\alpha_n-1}\right).
\end{align}
Taking $n\to \infty$ and using the definition of $a_n$,
\begin{align}
	\liminf_{n\to \infty} - \frac{1}{n a_n^2} \log \left(  \frac12\mathds{E}_{\breve{\mathcal{C}}^n} \left\|{\rho}_{B^n}^{\breve{\mathcal{C}}^n} - \breve{\rho}_{B^n} \right\|_1\right) \geq \frac{1}{2 \breve{V}(X{\,:\,}B)_\rho},
\end{align}
which proves the first claim in Proposition~\ref{prop:moderate_cc}. The second claims in Proposition~\ref{prop:moderate_cc} and Remark~\ref{remark:moderate_cc} follows similarly by recalling Theorems~\ref{theo:achievability_iid_oneshot} \& \ref{theo:sc}, Proposition~\ref{prop:sc_cc}, and Lemma~\ref{lemm:updiff} (see also the derivations given in \cite{CH17, CHD+21, SGC22a}).
\end{proof}

\section{Conclusions} \label{sec:conclusions}
In this work, we establish achievability and  strong converse for quantum soft covering using the random i.i.d.~codebook and the random constant composition codebook with the codebook size being fixed. In both settings, we obtain exponential convergence of the trace distance to $0$ or respectively to $1$, which measure the closeness between the codebook-induced state and the true marginal state. As a consequence, our results in achievability and strong converse combined implies that the optimal rate of quantum soft covering is the quantum mutual information $I(X{\,:\,}B)_\rho$.
We remark that our results hold for every blocklength $n\in\mathds{N}$, providing a large deviation analysis when the operating rate is fixed \cite{Hao-Chung, CHDH2-2018, CGH18, CHT19, Cheng2021a, Cheng2021b, LY21a, LY21b, SGC22a}.
Our results also extend to the moderate deivation regime when the rates approaches $I(X{\,:\,}B)_\rho$ moderately quickly \cite{CH17, CTT2017}.
Lastly, it is interesting to note that the sandwiched R\'enyi information $I^*\alpha(X{\,:\,}B)_\rho$ used in the exponent appears in classical-quantum channel coding as well \cite{WWY14, MO17, LY21b, Hay2112, SGC22a}, while the sandwiched Augustin information $\breve{I}^*\alpha(X{\,:\,}B)_\rho$ has appeared in other contexts using constant composition codes \cite{DW14, CGH18, MO18, CHDH2-2018}.

\section*{Acknowledgement}
H.-C.~Cheng would like to thank Bar\i\c{s} Nakibo\u{g}lu for discussions.
H.-C.~Cheng is supported by the Young Scholar Fellowship (Einstein Program) of the Ministry of Science and Technology in Taiwan (R.O.C.) under Grant MOST 110-2636-E-002-009, and are supported by the Yushan Young Scholar Program of the Ministry of Education in Taiwan (R.O.C.) under Grant NTU-110V0904,  Grant NTU-CC-111L894605, and Grand NTU-111L3401.


\section*{Appendix: Complex Interpolation and Noncommutative {$L_p$} Spaces} \label{sec:interpolation}

In this section, we briefly review the definition of the complex interpolation. We refer to \cite{BL76} for a detailed account of interpolation spaces.  Let $X_0$ and $X_1$ be two Banach spaces. Assume that there exists a Hausdorff topological vector space $X$ such that $X_0, X_1\subset X$ as subspaces. Let $\mathcal{S}=\{z\,|0\le \textsf{Re} (z)\le 1\}$ be the unit vertical strip on the complex plane, and $\mathcal{S}_0=\{z\,|0< \textsf{Re} (z)< 1\}$ be its open interior. Let $\F(X_0, X_1)$ be the space of all functions $f:\mathcal{S}\to X_0+X_1$, which are bounded and continuous on $\mathcal{S}$ and analytic on $\mathcal{S}_0$, and moreover
\[\{f(\mathrm{i}t)\pl|\pl t\in \mathbb{R}\}\subset X_0\pl ,\pl \{f(1+\mathrm{i}t)\pl|\pl t\in \mathbb{R}\}\subset X_1\pl.\]
$\F(X_0, X_1)$ is again a Banach space equipped with the norm
\[\norm{f}{\F} :=\max\left\{\pl \sup_{t\in \mathbb{R}} \norm{f(\mathrm{i}t)}{X_0}\pl,\pl \sup_{t\in \mathbb{R}}\norm{f(1+\mathrm{i}t)}{X_1}\right\}\pl. \]
The complex interpolation space $(X_0,X_1)_\theta$, for $0\le \theta\le 1$, is the quotient space of $\F(X_0,X_1)$ as follows,
\[(X_0, X_1)_\theta=\{\pl x\in X_0+X_1\pl| \pl x=f(\theta) \pl \text{for some }\pl  f\in \F(X_0, X_1)\pl\} \pl.\]
where quotient norm is
\begin{align}\label{def}\norm{x}{\theta}=\inf \{\pl\norm{f}{\F}\pl| \pl f(\theta)=x \pl\}\pl .\end{align}
It is clear from the definition that $X_0=(X_0,X_1)_0,  X_1=(X_0,X_1)_1$. For all $0<\theta<1$,
$(X_0,X_1)_\theta$ are called interpolation space of $(X_0,X_1)$.

The most basic example is that the $p$-integrable function spaces $L_p(\Omega,\mu)$ of a positive measure space $(\Omega,\mu)$. $L_p(\Omega,\mu)$ for $1\le p\le \infty$ forms a family of interpolation spaces, i.e.
\begin{align}\label{eq:in} L_p(\Omega,\mu)\cong [L_{p_0}(\Omega,\mu),{L_{p_1}}(\Omega,\mu)]_\theta \end{align}
holds isometrically for all $1\le p_0,p_1,p\le \infty, 0\le \theta\le 1$ such that
$\frac{1}{p}=\frac{1-\theta}{p_0}+\frac{\theta}{p_1}$. For a von Neumann algebra $(\mathcal{M},\text{Tr})$ equipped with normal faithful semifinite trace $\text{Tr}$,
the noncommutative $L_p$-norm is defined as $\norm{x}{p}=\text{Tr}(|x|^p)^{\frac{1}{p}}$ and $L_p(\mathcal{M},\text{Tr})$ (or shortly $L_p(\mathcal{M})$) is the completion of $\{x\in \mathcal{M} \pl |\pl \norm{x}{p}<\infty\}$. The noncommutative analog of \eqref{eq:in} is that
\begin{align}\label{eq:in2} L_p(\mathcal{M},\text{Tr})\cong [L_{p_0}(\mathcal{M},\text{Tr}),{L_{p_1}}(\mathcal{M},\text{Tr})]_\theta .\end{align}
In particular, the Schatten-$p$ class on a Hilbert space $\mathcal{H}$ are the $L_p$ spaces of $(\mathcal{B(H)},\tr)$ which satisfies
\[S_p(\mathcal{H}) \cong \left[S_{p_0}(\mathcal{H}), S_{p_1}(\mathcal{H}) \right]_{\theta}\pl.\]
Here $S_\infty(\mathcal{H})$ is identified with $\mathcal{B(H)}$. The complex interpolation relation has been already used in many works in quantum information theory, e.g.~\cite{WWY14,CGH19}. In this work, we will use
 the complex interpolation for a mixture of \eqref{eq:in} and \eqref{eq:in2}.
   For an operator-valued function $f:\Omega \to \mathcal{B(H)}$, its $L_p$ norm is given by
 \[ \norm{f}{L_p(\Omega, S_p(\mathcal{H}))}:=\left(\int_{\Omega} \norm{f(\omega)}{S_p}^p \d\omega\right)^{\sfrac{1}{p}}.\]
$L_p(\Omega, S_p(\mathcal{H}))$ is exactly the $L_p$-space of semi-classical system $L_\infty(\Omega,\mathcal{B(H)})$, which is a von Neumann algebra equipped with the trace $\tau(f)=\int_{\Omega}\tr(f(\omega))d\omega$). Thus $L_p(\Omega, S_p(\mathcal{H}))$ satisfies complex interpolation by \eqref{eq:in2} ($L_\infty(\Omega,\mathcal{B(H)})$. In particular,
$L_2(\Omega, S_2(\mathcal{H}))$ is a Hilbert space with inner product.
\[ \langle f, g\rangle = \int_{\Omega}\tr(f(\omega)^*g(\omega))\d\mu(\omega).\]

One widely used property of interpolation space is the following Riesz–Thorin interpolation theorem.

\begin{theo}[Riesz--Thorin interpolation theorem] \label{thm:interpolation} Let $(X_0, X_1)$ and $(Y_0, Y_1)$ be two compatible couples of Banach spaces and let $(X_0, X_1)_\theta$ and $(Y_0, Y_1)_\theta$ be the corresponding interpolation space of exponent $\theta$. Suppose $T:X_0 + X_1 \to Y_0 + Y_1$, is a linear operator bounded from $X_j$ to $Y_j$, $j=0, 1$. Then $T$ is bounded from $(X_0, X_1)_\theta$ to $(Y_0, Y_1)_\theta$, and moreover
\[ \|T:(X_0, X_1)_\theta\to (Y_0, Y_1)_\theta\|_{}\leq \|T:X_0\to Y_0\|_{}^{1-\theta }\|T:X_1\to Y_1\|_{}^{\theta}.\]
\end{theo}

\bibliographystyle{myIEEEtran}
\bibliography{reference.bib}

\end{document}